\documentclass[prl,10pt,twocolumn,showpacs,preprintnumbers,amsmath,amssymb,aps,superscriptaddress, longbibliography]{revtex4-1}
\pdfoutput=1
\usepackage{amsthm}
\usepackage{latexsym}
\usepackage{mathrsfs}
\usepackage{eucal}
\usepackage{tipa}
\usepackage{hyperref}

\allowdisplaybreaks[2]

\theoremstyle{plain}
\newtheorem{theorem}{Theorem}
\newtheorem{proposition}[theorem]{Proposition}

\newtheorem{corollary}[theorem]{Corollary}

\newtheorem{remark}[theorem]{Remark}

\DeclareMathOperator{\tho}{\text{\textthorn}}
\DeclareMathOperator{\edt}{\eth}

\def\Im{\textrm{Im}}
\def\Re{\textrm{Re}}

\def\InvSymb{\mathbb{I}}

\def\Rotation{\mathbf{a}}
\def\Mass{\mathbf{M}}
\def\Nut{\mathbf{N}}
\def\Cpara{\mathbf{c}}

\begin{document}

\preprint{}

\title[Gauge invariants on Kerr]{All local gauge invariants for perturbations of the Kerr spacetime}

\author{Steffen Aksteiner}
\email{steffen.aksteiner@aei.mpg.de}
\affiliation{Albert Einstein Institute, Am M\"uhlenberg 1, D-14476 Potsdam,
Germany}
\affiliation{Department of Mathematics, Princeton University, Princeton, NJ 08544-1000 USA}
\author{Thomas B\"ackdahl}
\email{thomas.backdahl@aei.mpg.de}
\affiliation{Albert Einstein Institute, Am M\"uhlenberg 1, D-14476 Potsdam,
Germany}
\affiliation{Mathematical Sciences, Chalmers University of Technology and University of Gothenburg, SE-412 96 Gothenburg, Sweden}

\date{\today}

\begin{abstract}
We present two complex scalar gauge invariants for perturbations of the Kerr spacetime defined covariantly in terms of the Killing vectors and the conformal Killing-Yano tensor of the background together with the linearized curvature and its first derivatives. These invariants are in particular sensitive to variations of the Kerr parameters. 
Together with the Teukolsky scalars and the linearized Ricci tensor they form a minimal set that generates all local gauge invariants. 
We also present curvature invariants that reduce to the gauge invariants in linearized theory.
\end{abstract}

% PACS, the Physics and Astronomy Classification Scheme.
%Use showkeys class option if keyword display desired
\maketitle

\addcontentsline{toc}{section}{Introduction}
\emph{Introduction.---}
Black hole perturbation theory plays a major role in numerical and analytical investigations of general relativity.
The coordinate freedom or diffeomorphism invariance of the Einstein equations implies gauge freedom for the linearized equations. For many applications it is essential to extract the gauge invariant content of the theory. The aim of this letter is to describe all local gauge invariants for perturbations of the rotating Kerr black hole.

The dynamics of gravitational perturbations of the Schwarzschild geometry is governed by the Regge-Wheeler \cite{regge:wheeler:1957} and Zerilli \cite{zerilli:1970} variables, see \cite{moncrief:1974, jezierski:1999} for a gauge invariant formulation, but also by the Bardeen-Press \cite{1973JMP....14....7B} variables. Gauge invariants of higher than second differential order have also been used in the literature, see e.g. \cite{2014PhRvL.112s1101D} for third order quantities and \cite{2016arXiv161108291S} for a set of higher order gauge invariants and their relations on Schwarzschild. The construction of the Bardeen-Press invariants has been generalized to the Kerr geometry by Teukolsky \cite{teukolsky:1973} and Wald showed in \cite{1973JMP....14.1453W} that one complex Teukolsky scalar determines the linearized gravitational field up to unphysical solutions and Kerr parameter variations. Motivated by the self-force problem, Merlin et. al. \cite{PhysRevD.94.104066} recently constructed three more real, third order scalar invariants and in \cite{2016arXiv160106084A} we found a third order gauge invariant vector field. Here we take a different perspective by asking: in terms of which variables can the gauge invariant content of the theory be described? As will be demonstrated in \cite{future:CompComplex}: there exist a finite number of invariants from which all invariants can be constructed by further differentiation and linear combination. This opens up the possibility of a systematic investigation of the field equations and differential relations satisfied by the gauge invariant quantities.

Linearized diffeomorphisms generated by a real vector field $\nu^a$ change the linearized metric $h_{ab}$ according to
\begin{align} \label{eq:hgauge}
h_{ab} \to h_{ab} + 2\nabla_{(a}\nu_{b)}.
\end{align} 
Depending on the background geometry, certain linear combinations of derivatives of $h_{ab}$ can be constructed to be independent of $\nu^a$ under the transformation \eqref{eq:hgauge}. Such quantities are called \textbf{local gauge invariants}, and play a fundamental role in black hole perturbation theory. Also non-local gauge invariants, often formulated in terms of separated modes or global integrals can be of interest, but here we restrict our attention to local quantities.

Stewart \& Walker \cite{1974RSPSA.341...49S} showed that any linearized field $\dot T$ around a given background, transforms under \eqref{eq:hgauge} with the Lie derivative along $\nu^a$ of its background value $T_0$,
\begin{align} \label{eq:Tgauge}
\dot T \to \dot T + \mathcal{L}_\nu T_0.
\end{align}
This in particular implies that the linearized Ricci tensor is gauge invariant on vacuum backgrounds. 

Any linear differential operator applied to a gauge invariant is also gauge invariant. Therefore, we say that a set of gauge invariants is \textbf{generating} if all gauge invariants can be expressed as a linear combinations of differential operators on elements of this set. For instance, the linearized curvature tensor is gauge invariant for perturbations of Minkowski space and one can show that it forms a generating set. 

We would like elements of the generating set to be of as low differential order as possible and also \textbf{minimal} in the sense that the removal of any element implies loss of information about the gauge invariant content, or equivalently that the set is no longer generating. 
Observe however, that the elements satisfy differential relations. For example, the set of components of the linearized curvature tensor on Minkowski space is minimal generating, but the elements are related by the differential Bianchi identities. 

For perturbations of the Kerr spacetime the two complex Teukolsky scalars (and their derivatives) together with the linearized Ricci tensor are well known gauge invariants. In this letter we add two complex scalar fields, $\InvSymb_\xi$ and $\InvSymb_\zeta$, to this list of local gauge invariants (collectively they can be found in \cite{PhysRevD.94.104066}, \cite{2016arXiv160106084A}, see remarks \ref{rem:ConnectionToImA} and \ref{rem:ConnectionToMerlinEtAl}, but here they are identified as a generating set for the first time). Their construction involves the Killing vectors $\xi^a, \zeta^a$, see proposition~\ref{prop:IvGI} below. The main result of this letter is the statement of a minimal generating set of gauge invariants in theorem~\ref{thm:CompleteSet}. 

Finally we present curvature invariants in the non-linear theory that reduce to constants on a Kerr background and to $\InvSymb_\xi$ and $\InvSymb_\zeta$ in linear theory. 

The proof that the set is generating is based on a universal compatibility complex of the Killing operator, see \cite{CompatComplexKhavkine}, and will be published separately in \cite{future:CompComplex}. For examples of the method on Lorentzian manifolds see e.g. \cite{2014arXiv1409.7212K}. The computations for this letter were performed with \emph{xAct} for \emph{Mathematica}, in particular using the package \emph{SpinFrames} \cite{SpinFrames}.

\medskip
\addcontentsline{toc}{section}{Geometry of Kerr}
\emph{Geometry of Kerr.---}
We use abstract index notation and let $g_{ab}$ denote the background Kerr metric with parameters $\Rotation$ and $\Mass$. Unless otherwise stated, frame dependent statements are valid in any principal Newman-Penrose tetrad $(l^a,n^a,m^a,\bar m^a)$ on the background. For clarity some coordinate expressions are given in Boyer-Lindquist like coordinates $(t,r,x=\cos\theta, \phi)$.

Let $\mathcal{Y}_{ab}$ be the anti-self dual conformal Killing-Yano tensor of the Kerr spacetime, see \cite{walker:penrose:1970CMaPh..18..265W}, normalized so that
\begin{align}
\xi^{c}={}&\tfrac{2}{3}i \nabla_{a}\mathcal{Y}^{ca}
\end{align}
is the real Killing vector $\partial_t$.
Furthermore, let
\begin{align}
\mathit{p} \equiv{}& \sqrt{\mathcal{Y}_{bd} \mathcal{Y}^{bd}}= r - i \Rotation x, &
U_{a} \equiv{}& - \nabla_{a}\log(\mathit{p}).
\end{align}
We base our construction on the 2-form $\mathcal{Y}_{ab}$ which in any principal tetrad takes the form
\begin{align}
\mathcal{Y}_{ab}={}&i \mathit{p} (l_{[a}n_{b]}
 -  m_{[a}\overline{m}_{b]}).
\end{align}
A second, linearly independent Killing vector is given by
\begin{align}
\zeta^a\equiv{}& 2 \mathcal{Y}^{ab} \overline{\mathcal{Y}}_{bc} \xi^{c} - \tfrac{1}{4} (\mathit{p}^2 + \overline{\mathit{p}}^2) \xi^{a}
=\Rotation^2 (\partial_t)^{a}
 + \Rotation (\partial_\phi)^{a},
\end{align}
see \cite{1977CMaPh..56..277C} for details. For later reference, we note the reality conditions
\begin{align}
\Psi_{2} \mathit{p}^3 ={}& \bar\Psi_{2} \overline{\mathit{p}}^3=-\Mass,&
\overline{U}^{a} \overline{\mathcal{Y}}_{ab}={}&- U^{a} \mathcal{Y}_{ab}=\tfrac{1}{2}i \xi_{b},\label{eq:RealityCond}
\end{align}
and the differential relation
\begin{align}
\nabla_{a}U_{b}={}&2 U_{a} U_{b}
 -  \tfrac{1}{2} g_{ab} (\Psi_{2}
 + 2 U_{c} U^{c})
 -  \mathit{p}^{-2} \xi_{a} \xi_{b}\nonumber\\
& + 2 \bar\Psi_{2} \mathit{p}^{-2} \mathcal{Y}_{(a}{}^{c}\overline{\mathcal{Y}}_{b)c}.\label{eq:CDeU}
\end{align}

\medskip
\addcontentsline{toc}{section}{Metric perturbations}
\emph{Metric perturbations.---}
To avoid complications with the tetrad gauge freedom, we treat metric perturbations covariantly in the style of \cite{BaeVal15} and denote the variation of a field $F$ by $\dot F$. Define the following version of the linearized Riemann tensor 
\begin{align}
\dot{R}_{abcd}
\equiv{}&2 g_{f[d} \nabla_{c]}\nabla_{[a}h_{b]}{}^{f}
 + \tfrac{2}{3} R_{[ab]}{}^{f}{}_{[c}h_{d]f}
  -  \tfrac{2}{3} R^{f}{}_{[ab][c}h_{d]f},
\end{align}
which is the mean value of variations with all indices up and all indices down. The spin-0 and spin-1 parts of the linearized curvature can be expressed by 
\begin{subequations}
\begin{align}
6 \mathit{p}^2 \vartheta \Psi_{2}={}&(\dot{R}_{b}{}^{d}{}_{cd} \mathcal{Y}_{a}{}^{c}
 -  \dot{R}_{abcd} \mathcal{Y}^{cd}
 -  \dot{R}_{acbd} \mathcal{Y}^{cd}) \mathcal{Y}^{ab},\\
4 \mathit{p}^2 \mathcal{Z}_{ab}={}&-2 \mathcal{Y}^{cd} \dot{R}_{[a|cd|}{}^{f}\mathcal{Y}_{b]f}
 + 3 \mathcal{Y}^{cd} \dot{R}_{[a}{}^{f}{}_{|cd|}\mathcal{Y}_{b]f}.
\end{align}
\end{subequations}
In any principal tetrad, we get
\begin{subequations}
\begin{align}
\vartheta \Psi_{0}={}&\dot{R}_{lmlm},\\
\vartheta \Psi_{1}={}&\tfrac{1}{2} \dot{R}_{lmln}
 -  \tfrac{1}{2} \dot{R}_{lmm\bar{m}},\\
\vartheta \Psi_{2}={}&\tfrac{1}{6} \dot{R}_{lnln}
 -  \tfrac{1}{3} \dot{R}_{lnm\bar{m}}
 + \tfrac{1}{3} \dot{R}_{lm\bar{m} n}
 + \tfrac{1}{6} \dot{R}_{m\bar{m} m\bar{m}},\\
\vartheta \Psi_{3}={}&\tfrac{1}{2} \dot{R}_{ln\bar{m} n}
 -  \tfrac{1}{2} \dot{R}_{m\bar{m} \bar{m} n},\\
\vartheta \Psi_{4}={}&\dot{R}_{\bar{m} n\bar{m} n},\\
\mathcal{Z}_{ab}={}&-2 \vartheta \Psi_{3} l_{[a}m_{b]}
 + 2 \vartheta \Psi_{1} \overline{m}_{[a}n_{b]}.
\end{align}
\end{subequations}
Here, $\vartheta \Psi_{i}$ are the components of the linearized Weyl spinor $\vartheta\Psi_{ABCD}$ introduced in \cite{BaeVal15}, but the formulas above can be used as definitions. They are related to linearized Newman-Penrose Weyl scalars $\dot{\Psi}_i$, but compensated for their linearized tetrad gauge dependence,
\begin{align}
\vartheta \Psi_{0} ={}& \dot\Psi_0, \,
\vartheta \Psi_{1} ={} \dot\Psi_1 + \tfrac{3}{2} \Psi_{2}( m_{a} \dot l^{a}
 - l_{a} \dot m^{a}),\,
\vartheta \Psi_{2} ={} \dot\Psi_2, \nonumber\\
\vartheta \Psi_{4} ={}& \dot\Psi_4, \,
\vartheta \Psi_{3} ={} \dot\Psi_3 - \tfrac{3}{2} \Psi_{2} ( n_{a} \dot{\bar{m}}^{a}
 -  \bar{m}_{a} \dot n^{a} ).
\end{align}

\medskip
\addcontentsline{toc}{section}{The set of gauge invariants}
\emph{The set of gauge invariants.---}
We can now construct third order gauge invariants.
\begin{proposition} \label{prop:IvGI}
Let $V^a$ be a real Killing vector field and
\begin{align}
\InvSymb_{V}={}&\mathit{p}^2 W^{a}  \nabla_{a}(\mathit{p}^4\vartheta \Psi_{2})
 -  \tfrac{1}{2} \Re(\mathit{p}^6 \vartheta \Psi_{2} \nabla_{a}W^{a})\nonumber\\
&-2i \Im(\mathit{p}^6 U^{a} W^{b}  \mathcal{Z}_{ab})
 -  \tfrac{3}{2} \mathit{p}^6\Psi_{2} U^{a} W^{b}  h_{ab},\label{eq:InvV}
\end{align}
where the vector field
$W_{a} \equiv 2i \mathit{p}^{-3} V^{b} \mathcal{Y}_{ab} $
is assumed to satisfy the condition
\begin{align}
\overline{\mathit{p}}^3 \overline{U}_{[a}\overline{W}_{b]}={}&- \mathit{p}^3 U_{[a}W_{b]}.\label{eq:OrthogonalityCond}
\end{align}
Then $\InvSymb_{V}$ is a local gauge invariant. 
\end{proposition}
\begin{proof}
A consequence of the Killing equation gives $\Psi_{2} U^{a} V_{a}=0$.
For a pure diffeomorphism we get
\begin{subequations}
\begin{align}
\vartheta \Psi_{2}={}&3 \Psi_{2} U^{a} \nu_{a} \label{eq:VarPsi2GaugeTensorEq1},\\
\mathcal{Y}_{b}{}^{c} \mathcal{Z}_{ac}={}&- \tfrac{3}{2} \Psi_{2} \mathcal{Y}_{b}{}^{c} \nabla_{[a}\nu_{c]}
 + \tfrac{3}{2} \Psi_{2} \mathcal{Y}_{a}{}^{c} \nabla_{[b}\nu_{c]}\nonumber\\
& - 3 \Psi_{2} U_{[a}\mathcal{Y}_{b]}{}^{c}\nu_{c}
 - 3 \Psi_{2} U^{c}\mathcal{Y}_{[a|c|}\nu_{b]}.
\end{align}
\end{subequations}
This gives
\begin{align}
 \mathit{p}^6U^{a} W^{b} \mathcal{Z}_{ab}={}&
 -  \tfrac{3}{2}\mathit{p}^6\Psi_{2} U^{a} W^{b}  (\nabla_{[a}\nu_{b]}
  + 2 U_{[a}\nu_{b]})
\nonumber\\
&-3i \mathit{p}^3\Psi_{2} U^{a} V^{b} \mathcal{Y}_{a}{}^{c}  \nabla_{[b}\nu_{c]}.
\end{align}
The conditions \eqref{eq:RealityCond} and \eqref{eq:OrthogonalityCond} then imply
\begin{align}
& i \Im(\mathit{p}^6 U^{a} W^{b}  \mathcal{Z}_{ab})=\nonumber\\
&\hspace{6ex}- \tfrac{3}{2} \Psi_{2} U^{a} W^{b} \mathit{p}^6 (\nabla_{[a}\nu_{b]}
 + U_{[a}\nu_{b]}
 + \overline{U}_{[a}\nu_{b]}).
\end{align}
Together with \eqref{eq:VarPsi2GaugeTensorEq1}, we get for a pure diffeomorphism
\begin{align}
\mathbb{I}_{V}{}={}&-3 \Psi_{2} U^{a} U^{b} W_{a} \mathit{p}^6 \nu_{b}
 -  \tfrac{3}{4} \Psi_{2} U^{a} \mathit{p}^6 \nu_{a} \nabla_{b}W^{b}\nonumber\\
& -  \tfrac{3}{4} \bar\Psi_{2} \overline{U}^{a} \overline{\mathit{p}}^6 \nu_{a} \nabla_{b}\overline{W}^{b}
 + 3 \Psi_{2} W^{a} \mathit{p}^6 \nu^{b} \nabla_{(a}U_{b)}\nonumber\\
& + 3 \Psi_{2} U^{a} W^{b} \mathit{p}^6 (U_{[a}\nu_{b]}
 + \overline{U}_{[a}\nu_{b]}).
\end{align}
The relations \eqref{eq:CDeU}, \eqref{eq:OrthogonalityCond}, \eqref{eq:RealityCond} imply
\begin{subequations}
\begin{align}
W^{b} \nabla_{b}U_{a}={}&
2 U_{a} U^{b} W_{b}
 -  U_{b} U^{b} W_{a}
 - \Psi_{2} \Re(W_{a}),\hspace{-.6ex}\\
U^{b} \nabla_{b}W_{a}={}& 
2 U^{b} U_{(a}W_{b)}
 - \Psi_{2} \Re(W_{a}),\\
\hspace{-.3ex}\overline{\mathit{p}}^3\overline{U}_{a}  \nabla_{b}\overline{W}^{b}={}&-4 \Psi_{2} \Re(W_{a}) \mathit{p}^3
 -  U_{a} \mathit{p}^3 \nabla_{b}W^{b} \nonumber\\
&+ 4 (U^{b}  -  \overline{U}^{b}) \mathit{p}^3 U_{[a}W_{b]}.
\end{align}
\end{subequations}
Together, we get $\InvSymb_{V}{}=0$ for a pure diffeomorphism.
\end{proof}
\begin{remark}
For the case $V^a=\xi^a$ we get $W^{a} = - U^{a} \mathit{p}^{-1}$,
and for the case $V^a=\zeta^a$ we get
\begin{align}
W_{a} \mathit{p}^3={}&- \tfrac{1}{2} \overline{U}_{a} \mathit{p}^2 \overline{\mathit{p}}^2
 + \tfrac{1}{4} U_{a} \mathit{p}^2 (\mathit{p}^2
 + \overline{\mathit{p}}^2).
\end{align}
Both of these vectors satisfy the required condition \eqref{eq:OrthogonalityCond}, so $\InvSymb_{\xi}$ and $\InvSymb_{\zeta}$ are gauge invariant.
\end{remark}

\begin{corollary} \label{cor:InvSet}
A set of local gauge invariant quantities for perturbations of the Kerr spacetime is given by
\begin{subequations}\label{eq:GIset}
\begin{align} 
&\text{Teukolsky scalars} &&\vartheta \Psi_{0}, \vartheta \Psi_{4}, \label{eq:GI1}\\
&\text{Linearized Ricci} &&\dot{R}_{ab}=\dot{R}_{acb}{}^{c}, \label{eq:GI2} \\
&\text{Killing invariants} &&\InvSymb_{\xi}{}, \InvSymb_{\zeta}. \label{eq:GI5}
\end{align}
\end{subequations}
\end{corollary}
Note that \eqref{eq:GI1}, \eqref{eq:GI2} depend on up to second derivatives of linearized metric, while \eqref{eq:GI5} depends on up to third derivatives.
\begin{theorem}[\cite{future:CompComplex}] \label{thm:CompleteSet}
The set of gauge invariants in corollary~\ref{cor:InvSet} is minimal and generates all local gauge invariants for perturbations of the Kerr spacetime with $\Rotation \neq 0$.
\end{theorem}
Arguments for minimality are given below and for a proof of the theorem we refer to \cite{future:CompComplex}.

It should be noted that a generating set of gauge invariants can degenerate if restricted to more special backgrounds, in the sense that certain components of the set can be derived from more elementary gauge invariants. Examples are the second order Regge-Wheeler variable on Schwarzschild and the linearized curvature components of spin-0 and 1 on Minkowski. 

Also the spherical Killing vectors on Schwarzschild satisfy condition \eqref{eq:OrthogonalityCond}, and therefore lead to gauge invariants. On the other hand, the Regge-Wheeler variable $\Im\vartheta\Psi_2$ is gauge invariant and hence certain real or imaginary parts of Killing invariants can be generated from it.

\begin{remark} \label{rem:ConnectionToImA}
In \cite{2016arXiv160106084A} we derived a covariant version of the Teukosky-Starobinski identities. In these identities a real, gauge invariant vector field $\Im\mathcal{A}^c$ appears naturally. It can be expressed in terms of the invariants \eqref{eq:GIset}, for example $\Im\mathcal{A}^{a} V_{a}={}- \frac{1}{81} \Im\InvSymb_V$ for both isometries in the source-free case. This partly initiated the systematic search for gauge invariants.
\end{remark}

\begin{remark} \label{rem:ConnectionToMerlinEtAl}
 Merlin et. al. found three real gauge invariants, \cite{PhysRevD.94.104066}, in a coordinate based construction. They are related to  $\Re \InvSymb_\xi $, $\Re\InvSymb_\zeta$ and $\Im \InvSymb_\xi $ via
\begin{subequations}\label{eq:MerlinInv}
\begin{align}
\mathcal{I}_{1}{}={}&- \frac{2 \bigl(Re(\InvSymb_\zeta) + r^2 Re(\InvSymb_\xi)\bigr) (r^2+ \Rotation^2 x^2)}{3 \Mass (r^2 - 2\Mass r + \Rotation^2)^2},\\
\mathcal{I}_{2}{}={}&- \frac{2 \bigl(Re(\InvSymb_\zeta) -  \Rotation^2x^2 Re(\InvSymb_\xi) \bigr) (r^2+ \Rotation^2 x^2)}{3 \Mass \Rotation^2 (1- x^2)},\\
\mathcal{I}_{3}{}={}&\frac{Im(\InvSymb_\xi) (r^2+ \Rotation^2 x^2)^2}{3 \Mass \Rotation (r^2 - 2\Mass r + \Rotation^2) \sqrt{1 - x^2}}.
\end{align}
\end{subequations}
\end{remark}

\medskip
\addcontentsline{toc}{section}{Type D variations and independence of gauge invariants}
\emph{Type D variations and independence of gauge invariants.---}
We use the Plebanski-Demianski solution \cite{1976AnPhy..98...98P} in vacuum. In coordinates $(t,r,x=\cos\theta, \phi)$ define the Newman-Penrose tetrad
\begin{subequations} 
\begin{align}
l ={}&\frac{(1 -  \Cpara r x)}{\sqrt{2\Sigma\Delta_{r}}}  
\big( (r^2+\Rotation^2)\partial_t +   \Delta_{r} \partial_r   + \Rotation \partial_\phi \big) ,\\
n ={}&\frac{(1 -  \Cpara r x)}{\sqrt{2\Sigma\Delta_{r}}}  
\big( (r^2+\Rotation^2)\partial_t -   \Delta_{r} \partial_r   + \Rotation \partial_\phi \big) ,\\
m ={}&\frac{(1 -  \Cpara r x)}{\sqrt{2\Sigma\Delta_{x}}} 
\big( i \Rotation (1-x^2)\partial_t  -\Delta_{x} \partial_x + i \partial_\phi \big),\\
\Delta_{x}{}={}&1 + 2 \Nut \Rotation^{-1} x -  x^2 + 2 \Cpara \Mass x^3 -  \Cpara^2 a^2 x^4,\\
\Delta_{r}{}={}& \Rotation^2 - 2 \Mass r + r^2 - 2 \Cpara \Nut \Rotation^{-1} r^3 -  \Cpara^2 r^4, \\
 \Sigma={}&r^2 + \Rotation^2 x^2, \label{eq:Deltar}
\end{align}
\end{subequations} 
with parameters $\Mass, \Nut, \Rotation, \Cpara$ for mass, NUT charge, angular momentum and c-metric, respectively. A variation in each of the parameters leads to specific values of the invariants showing their functional independence.

For pure mass  ($\dot \Mass$) and angular momentum ($\dot \Rotation$) perturbations, the invariants take the form
\begin{subequations}
\begin{align}
\InvSymb_\xi={}&\dot \Mass,&
\InvSymb_\zeta={}&2 \Rotation^2 \dot \Mass 
 - 3 \Mass \Rotation \dot \Rotation,
\end{align}
while perturbing in direction of the NUT ($\dot \Nut$) yields
\begin{align}
\InvSymb_\xi={}&-i \dot \Nut
 + \frac{2i \Mass }{\overline{\mathit{p}}}\dot \Nut,\\
\InvSymb_\zeta={}&-i \Rotation^2 \dot \Nut
 + \Rotation x (r - 2 \Mass 
 -  \frac{\Mass \mathit{p}}{\overline{\mathit{p}}}) \dot \Nut,
\end{align}
and perturbing in the c-metric direction ($\dot \Cpara$) gives
\begin{align}
\InvSymb_\xi={}&\frac{6 \Mass^2 r x}{\overline{\mathit{p}}}\dot \Cpara 
 + 3 \Mass  \bigl(i \Rotation 
 + (\Mass -  r) x\bigr)\dot \Cpara,\\
\InvSymb_\zeta={}&\frac{6 \Mass^2 \Rotation^2  r x^3}{\overline{\mathit{p}}}\dot \Cpara 
 - 3i \Mass \Rotation  (\mathit{p}^2
 -  r^2 x^2)\dot \Cpara. \label{eq:IinvzetaVarcmetric}
\end{align}
\end{subequations}
Observe that $\InvSymb_\xi,\InvSymb_\zeta$ are real for $\dot \Mass, \dot \Rotation$ perturbations, but complex for $\dot \Nut$ and $\dot \Cpara$ perturbations. 
From the explicit form above we conclude that the four real degrees of freedom of $\InvSymb_\xi,\InvSymb_\zeta$ are functionally independent. 
Furthermore, there are algebraically special frequency solutions, see e.g. \cite{chandrasekhar:MR1210321}, turning on only one of $\vartheta \Psi_{0}$, $\vartheta \Psi_{4}$. 
Similarly, metric perturbations turning on specific components of the linearized Ricci tensor are possible to construct by a linearized conformal transformation. Hence, we have a sequence of solutions turning on one invariant after the other. This motivates why all 18 invariants are needed on Kerr with $\Rotation \neq 0$. Even though the gauge invariants are independent in this way, they will satisfy a set of differential compatibility equations. These relations will be stated and used in \cite{future:CompComplex} for the proof of theorem \ref{thm:CompleteSet}.

One can argue that components of the linearized curvature are the only possible gauge invariants of second order, and that no gauge invariant curvature component carries the $\dot \Mass, \dot \Rotation$, $\dot \Nut$ and $\dot \Cpara$ perturbations. The differential order of $\InvSymb_\xi,\InvSymb_\zeta$ is therefore minimal.

\medskip
\addcontentsline{toc}{section}{GHP form of gauge invariants}
\emph{GHP form of gauge invariants.---}
In a principal tetrad the Killing invariants take the GHP \cite{GHP} form 
\begin{align}
\InvSymb_\xi={}&- \mathit{p} (\rho'\tho
 + \rho\tho'
 -  \tau'\edt
 -  \tau\edt' )(\mathit{p}^4\vartheta \Psi_{2})
 -  \tfrac{1}{2} \Psi_{2} \mathit{p}^5 \vartheta \Psi_{2}\nonumber\\
& -  \tfrac{1}{2} \bar\Psi_{2} \overline{\mathit{p}}^5 \overline{\vartheta \Psi}_{2}
 + \tfrac{3}{2} \Psi_{2} \mathit{p}^5 (h_{nn} \rho^2
 + 2 h_{ln} \rho \rho'
 + h_{ll} \rho'^2\nonumber\\
&{}\hspace{2ex}{} - 2 h_{n\bar{m}} \rho \tau
 - 2 h_{l\bar{m}} \rho' \tau
 + h_{\bar{m} \bar{m}} \tau^2
 - 2 h_{nm} \rho \tau'\nonumber\\
&{}\hspace{2ex}{} - 2 h_{lm} \rho' \tau'
 + 2 h_{m\bar{m}} \tau \tau'
 + h_{mm} \tau'^2),
\end{align}
and with $\mathit{p}_+ = \mathit{p} + \overline{\mathit{p}}, \mathit{p}_- = \mathit{p} - \overline{\mathit{p}}$,
\begin{align}
\InvSymb_\zeta={}&\tfrac{1}{4} \mathit{p} \big( \mathit{p}_-^2 (\rho'\tho
 + \rho\tho' ) -  \mathit{p}_+^2 (\tau'\edt
  + \tau\edt' ) \big)(\mathit{p}^4\vartheta \Psi_{2})\nonumber\\
& + \tfrac{1}{4} \Re\Bigl(\mathit{p}^5 \vartheta \Psi_{2} \bigl(
 \Psi_{2} \mathit{p}_+\mathit{p}_- -2 \bar\Psi_{2} \overline{\mathit{p}}^2
 - 4 \mathit{p} (\mathit{p}_- \rho \rho' \nonumber\\
 &{}\hspace{2ex}{} -\mathit{p}_+ \tau \tau')\bigr)\Bigr)
 + 2i \Im\bigl(\mathit{p}^6 \overline{\mathit{p}} (\vartheta \Psi_{3} \rho \tau
 + \vartheta \Psi_{1} \rho' \tau')\bigr)\nonumber\\
 & -  \tfrac{3}{8} \Psi_{2} \mathit{p}^5 \bigl( 
 \mathit{p}_-^2 (h_{nn} \rho^2  +  2 h_{ln} \rho\rho'  + h_{ll} \rho'^2)\nonumber\\
 &{}\hspace{2ex}{} -2 \mathit{p}_+\mathit{p}_- (h_{n\bar{m}} \rho \tau
   + h_{l\bar{m}} \rho' \tau
  + h_{nm} \rho \tau' + h_{lm} \rho' \tau')\nonumber\\
 &{}\hspace{2ex}{} + \mathit{p}_+^2 (h_{\bar{m} \bar{m}} \tau^2
  +  2 h_{m\bar{m}} \tau\tau' + h_{mm} \tau'^2) \bigr).
\end{align}

\medskip
\addcontentsline{toc}{section}{Curvature invariants}
\emph{Curvature invariants.---}
The scalars $\InvSymb_\xi, \InvSymb_\zeta$ can be derived from linearizations of tensors built from the curvature and its derivatives in the full theory. On a general vacuum spacetime with anti-self dual Weyl curvature $\mathcal{C}_{abcd}=\tfrac{1}{2} C_{abcd} + \tfrac{1}{2}i {}^*C_{abcd}$ define the curvature invariants
\begin{align} 
\widehat{\mathcal{I}}={}&\tfrac{1}{24} \mathcal{C}_{abcd} \mathcal{C}^{abcd},\label{eq:Iinv}&
\mathcal{I}={}& \widehat{\mathcal{I}}^{1/6}.
\end{align}
Furthermore, define the complex curvature invariant
\begin{align} \label{eq:MassInvariant}
\mathbb{M}={}& -\mathcal{I}^{-4} (\nabla_a \mathcal{I})(\nabla^a \mathcal{I}) + \mathcal{I}  + \overline{\mathcal{I}}.
\end{align}
On a Kerr spacetime in a principal tetrad we find
$\hat{\mathcal{I}} = \Psi_2^2$ and $\mathbb{M}$ turns out to be the real constant $\mathbb{M}_{\text{Kerr}} =  \Mass^{- 2/3}$.
Due to \eqref{eq:Tgauge}, it follows that the variation of \eqref{eq:MassInvariant} around Kerr is gauge invariant and a lengthy calculation shows
\begin{align}
 \dot{\mathbb{M}} ={}& - \tfrac{2}{3} \Mass^{-5/3} \InvSymb_\xi.
\end{align}
Similarly define the real scalar curvature invariant
\begin{align}
\mathbb{A}={}&
 \frac{-1}{\vert\mathcal{I}\vert^{4}} \Bigl(\nabla_a \frac{\mathcal{I}}{\overline{\mathcal{I}}}\Bigr) \Bigl(\nabla^a \frac{\overline{\mathcal{I}}}{\mathcal{I}}\Bigr)
-2 \Im\Bigl(\frac{1}{\mathcal{I}^{2}}\Bigr) \Im(2 \mathcal{I}
 -  \mathbb{M}).
\end{align}
In the background it reduces to $\mathbb{A}_{\text{Kerr}} =4 \Rotation^2  \Mass^{-4/3}$. The variation of $\mathbb{A}$ around a Kerr background shows
\begin{align}
\dot{\mathbb{A}}={}&- \tfrac{8}{3} \Mass^{- 7/3} \Re\InvSymb_\zeta.
\end{align}
To express  $\Im \InvSymb_\zeta$, define the real, symmetric, trace-free two tensor
\begin{align}
T_{ac}={}&\Im \left( \mathcal{I}^8 \overline{\mathcal{I}}^5 \overline{\mathcal{C}}_{(a}{}^{b}{}_{c)}{}^{d}(\nabla_{b}\mathcal{I})(\nabla_{d}\overline{\mathcal{I}}) \right),
\end{align}
which equals \cite[eq.(18)]{ferrando:saez:2009CQGra..26g5013F} up to a constant. 
On a general type D spacetime it has the non-trivial factor $\bar\rho \tau + \rho \bar\tau'$ which is zero in the Kerr case. Variation around Kerr and contraction into $\xi^a$ leads to
\begin{align}
\dot T_{ac} \xi^{a} \xi^{c} ={}&\frac{\Psi_{2}^5 \mathit{p}^4}{\overline{\mathit{p}}^{11}} 
\left( \xi^c\zeta_c \Im\InvSymb_\xi - \xi^c \xi_c \Im \InvSymb_\zeta \right).
\end{align}

\medskip
\addcontentsline{toc}{section}{Conclusions}
\emph{Conclusions.---}
In this letter we introduced two complex scalar gauge invariants for perturbations of Kerr spacetime. Together with the Teukolsky scalars and the Ricci tensor they form a minimal generating set of 18 real scalar invariants. A similar construction on a Schwarzschild background leads to a set of 19 real scalar invariants and for Minkowski space, it is known to consist of the 20 real scalar components of the linearized curvature tensor, see e.g. \cite{2014arXiv1409.7212K}. 
Whether there is a relation between the minimal number of generators for gauge invariants and the number of parameters of the background, also in other spacetimes, is yet unclear.

We would also like to point out that the invariants as defined in \eqref{eq:InvV} directly depend on the background isometries. The alternative definition in terms of curvature invariants does not make explicit use of this structure and may be interesting for tracking type D parameter variations in numerical evolution as well as for higher order self-force problems.
It would also be interesting to analyze the set of gauge invariants from the perspective of the black hole stability problem. Assuming that the Teukolsky scalars are under control, relations to the other invariants can be analyzed without gauge fixing and yield additional flexibility for the integration of the remaining field equations after gauge fixing. 

The geometric background and the full proof of theorem~\ref{thm:CompleteSet}, will be given in~\cite{future:CompComplex}.

\medskip
\addcontentsline{toc}{section}{Acknowledgements}
\emph{Acknowledgements.---}
We thank B. Whiting, L. Andersson and I. Khavkine for many enlightening discussions and in particular B. Whiting for bringing the problem of finding additional gauge invariants to our attention. We thank M. van~de Meent for pointing out the paper \cite{PhysRevD.94.104066}, and for discussions. S.A. thanks Princeton University for hospitality and financial support.
\vspace{-3ex}

\newcommand{\arxivref}[1]{\href{http://www.arxiv.org/abs/#1}{{arXiv.org:#1}}}
\newcommand{\mnras}{Monthly Notices of the Royal Astronomical Society}

%\bibliography{gr6}

\bigskip 

\end{document}